\documentclass[11pt]{article}

\usepackage{amssymb}
\usepackage[cmex10]{amsmath}
\usepackage{algorithm}
\usepackage[noend]{algpseudocode}
\usepackage[pdftex]{graphicx}
\usepackage{verbatim}%for comment blocks
\usepackage{url}
\usepackage{cite}
\usepackage{subfig}
\usepackage{multirow}
\usepackage{multicol}
\usepackage{theorem,latexsym,graphicx}
\usepackage{amsmath,amssymb,enumerate}
\usepackage{xspace}
\usepackage{tikz}
\usepackage{authblk}

\usepackage[top=1in, bottom=1in, left=1in, right=1in]{geometry}

\newtheorem{lemma}{Lemma}
\newtheorem{theorem}[lemma]{Theorem}

\newtheorem{definition}[lemma]{Definition}

\newtheorem{proposition}[lemma]{Proposition}

\newenvironment{proof}{\vspace{-0.05in}\noindent{\bf Proof:}}%
{\hspace*{\fill}$\Box$\par\vspace{0.05in}}

\newcommand{\ignore}[1]{}

\usepackage{url}

\newcommand{\scrit}{\textsf{$s_{crit}$}}
\newcommand{\J}{\textsf{$\mathcal{J}$}}

\begin{document}

%\mainmatter % start of an individual contribution

\title{A Fully Polynomial-Time Approximation Scheme for Speed Scaling
with Sleep State}

\author[1]{Antonios Antoniadis}
\author[2]{Chien-Chung Huang}
\author[1]{Sebastian Ott}
\affil[1]{Max-Planck-Institut f\"ur Informatik\\
  Saarbr\"{u}cken, Germany\\
  \texttt{\{aantonia,ott\}@mpi-inf.mpg.de}}
\affil[2]{Chalmers University\\
  G\"{o}teborg, Sweden\\
  \texttt{villars@gmail.com}}

\date{}
\maketitle

\abstract{ We study classical deadline-based preemptive scheduling of tasks in a computing
environment equipped with both dynamic speed scaling and sleep state
capabilities: Each task is specified by a release time, a deadline
and a processing volume, and has to be scheduled on a single,
speed-scalable processor that is supplied with a sleep state. In the
sleep state, the processor consumes no energy, but a constant wake-up
cost is required to transition back to the active state. In
contrast to speed scaling alone, the
addition of a sleep state makes it sometimes beneficial to accelerate
the processing of
tasks in order to transition the processor to the
sleep state for longer amounts of time and incur further energy savings. The goal is to output a feasible schedule
that minimizes the energy consumption. Since the introduction of the
problem by Irani et al.~\cite{ISG}, its exact computational complexity
has been repeatedly posed as an open question (see e.g.~\cite{IP05,AA12,BCD12}). The
currently best known upper and lower bounds are a $4/3$-approximation
algorithm and NP-hardness due to~\cite{AA12} and~\cite{AA12,KS13},
respectively. 

We close the aforementioned gap between the upper and lower bound on the computational complexity of
speed scaling with sleep state by presenting a fully polynomial-time
approximation scheme for the problem. The scheme is based on a
transformation to a non-preemptive variant of the problem, and a
discretization that exploits a carefully
defined lexicographical ordering among schedules.

}

\section{Introduction}

As energy-efficiency in computing environments becomes more and more
crucial, chip manufacturers are increasingly incorporating
energy-saving functionalities to their processors. One of the most common
such functionalities is \emph{dynamic speed scaling}, where the processor is
capable to dynamically adjust the speed at which it operates. A higher speed implies a
higher performance, but this performance comes at the cost of a higher energy
consumption. On the other hand, a lower speed results in better
energy-efficiency, but at the cost of performance degradation. In
practice, it has been observed~\cite{Brooks00,BCKP12} that the power consumption of the processor is approximately
proportional to its speed cubed. However, even when the processor is
idling, it consumes a non-negligible amount of energy just for the sake
of ``being active'' (for example because of leakage
current). Due to this fact, additional
energy-savings can be obtained by further incorporating a \emph{sleep state} to the
processor, in addition to the speed-scaling capability. A sleep state is a state of negligible or even zero
energy-consumption, to which the processor can transition
when it is idle. Some fixed energy-consumption is then required to
transition the processor back to the active state in order to
continue processing.

This article studies the offline problem of minimizing energy-consumptions in
computational settings that are equipped with both speed scaling and
sleep state capabilities. This problem is called \textit{speed scaling with sleep state},
and the algorithmic study of it was initiated 
in~\cite{ISG}. 

Consider a processor that is equipped with
two states: the \textit{active state} during which it can execute jobs
while incurring some energy consumption, and the \textit{sleep state}
during which no jobs can be executed, but also no energy is
consumed. We assume that a \textit{wake-up operation}, that is a
transition from the sleep state to the active state, incurs a
constant energy cost $C>0$, whereas transitioning from the active state to the sleep state
is free of charge. Further, as in~\cite{AA12, ISG}, the power required by the processor in the
active state is a convex and non-decreasing function $P$ of its speed
$s$. We assume that $P(0)>0$, since (i) as already mentioned,
real-world processors are known to have leakage
current and (ii) otherwise the sleep state would be
redundant. 
Further motivation for considering arbitrary convex power functions
for speed scaling can be found, for example, in ~\cite{BCP13}.% [Yao et al. 1995 is already for arbitrary convex functions (offline problem), this paper is for flow time]

The input is a set \J\ of $n$ jobs. % $j_1,j_2,\dots ,j_n$. 
Each job $j$ is
associated with a release time $r_j$, a deadline $d_j$ and a
processing volume $v_j$. One can think of the processing volume as the
number of CPU cycles that are required in order to completely process
the job, so that if job $j$ is processed at a speed of $s$, then $v_j/s$
time-units are required to complete the job. We call the interval
$[r_j,d_j)$ the \emph{allowed interval} of job $j$, and say that job
$j$ is active at time point $t$ if and only if $t\in
[r_j,d_j)$\footnote{Unless stated differently, throughout the text an interval will always have
  the form $[x_1,x_2)$.}.
Furthermore, we may assume without loss of generality
that $v_{min}:= \min_{j\in\J}v_j$ is normalized to $1$, and that
$\min_{j\in \J}r_j=0$. Further, let $d_{max}:=\max_{j\in \J} d_j$ be the last deadline of the job set $\J$.

 A \emph{schedule} is defined as a mapping of every time point $t$ to
 the state of the processor, its speed, and the job being processed at $t$ (or \emph{null} if there is no job running at $t$). Note that
the processing speed is zero whenever the processor sleeps, and that a job is only processed when the speed is strictly positive. A schedule is
called \emph{feasible} when the whole processing volume of every task
$j$ is completely processed in $j$'s allowed interval $[r_j,d_j)$. Preemption of jobs is
allowed. 

The energy consumption incurred by schedule $\mathcal{S}$ while the processor is in the
active state, is its power integrated over time, i.e. $\int P(s(t)) dt$, where $s(t)$ is the processing speed at time $t$, and the integral is taken
over all time points in $[0,d_{max})$ during which the processor is
active under $\mathcal{S}$. Assume that $\mathcal{S}$ performs $k$
transitions from the sleep state to the active state. (We will assume
throughout the paper that initially, prior to the first release time,
as well as finally, after the last deadline, the processor is in the active
state. However, our results can be easily adapted for the setting where the
processor is initially and/or eventually in the sleep state). Then the total
energy consumption of $\mathcal{S}$ is $E(S) := \int P(s(t))dt +kC$, where
again the integral is taken over all time points at which
$\mathcal{S}$ keeps the processor in the active state. 
We are seeking a
feasible schedule that minimizes the total energy consumption. 

Observe that, by Jensen's inequality,
and by the convexity of the power function, it is never beneficial to
process a job with a varying speed. Irani et al. \cite{ISG} observed
the existence of a \emph{critical speed \scrit}, which is the
most efficient speed for processing tasks. This critical speed
is the smallest speed that minimizes the function $P(s)/s$. Note that,
by the convexity of $P(s)$, the only case where the critical speed
$\scrit$ is not well defined, is when $P(s)/s$ is always
decreasing. However, this would render the setting unrealistic, and
furthermore make the algorithmic
problem trivial, since it would be optimal to process
every job at an infinite speed. We may therefore assume that this case
does not occur. Further, it can be shown~(see \cite{ISG}) that for any
$s\ge \scrit$, the function $P(s)/s$ is non-decreasing.

\subsection{Previous Work}

The theoretical model for dynamic speed scaling was introduced in a seminal paper by Yao,
Demers and Shenker~\cite{YDS}. They developed a polynomial time
algorithm called \emph{YDS}, that outputs a minimum-energy schedule
for this setting. Irani, Shukla and Gupta~\cite{ISG} initiated the algorithmic study of speed scaling combined with a sleep state. Such a
setting suggests the so-called \emph{race to idle} technique
where some tasks are accelerated over their minimum required speed in
order to incur a higher energy-saving by transitioning the processor
to the sleep state for longer periods of time (see~\cite{Bailis11,
  Gandhi09, Garrett07, Raghavan13} and references therein for more
information regarding the race to idle technique). Irani et al. developed
a $2$-approximation algorithm for speed scaling with sleep state, but the computational
complexity of the scheduling problem has remained open. The first step
towards attacking this open problem was made by Baptiste~\cite{B06}, who
gave a polynomial time algorithm for the case when the processor
executes all tasks at one fixed speed level, and all tasks are of
unit-size. Baptiste's algorithm is based on a clever dynamic
programming formulation of the scheduling problem, and was later
extended to (i) arbitrarily-sized tasks in~\cite{BCD12}, and (ii) a
multiprocessor setting in~\cite{DGHSZ13}.

More recently, Albers and Antoniadis~\cite{AA12} improved the upper
bound on the approximation ratio of the general problem, by developing a
$4/3$-approximation algorithm. With respect to the lower
bound, \cite{AA12} give an NP-hardness reduction from the
\emph{partition} problem. The reduction uses a particular
power function that is based on the partition instance, i.e., it is
considered that the power function is part of the input. The reduction of~\cite{AA12}
was later refined by Kumar and Shannigrahi~\cite{KS13}, to show that the problem
is NP-hard for any fixed, non-decreasing and strictly convex power
function. 

The online setting of the problem has also been studied. Irani et al.~\cite{ISG} give a $(2^{2\alpha-2}\alpha^\alpha +
2^{\alpha -1+2})$-competitive online
algorithm. Han et al. \cite{HLLTW10} improved upon this result by
developing an $(\alpha^\alpha+2)$-competitive algorithm for the
problem. Both of the above results assume a power function of the form
$ P(s) = s^\alpha + \beta$, where $\alpha>1$ and $\beta>0$ are constants.

A more thorough discussion on the above scheduling problems can be
found in the surveys~\cite{A10,IP05}.

%The online setting of dynamic speed scaling without sleep state has also been extensively
%studied, see~\cite{YDS,BKP07} and the references within.

\subsection{Our Contribution}
We study the offline setting of speed scaling with sleep state. Since the introduction of the
problem by Irani et al.~\cite{ISG}, its exact computational complexity
has been repeatedly posed as an open question (see e.g.~\cite{IP05,AA12,BCD12}). The
currently best known upper and lower bounds are a $4/3$-approximation
algorithm and NP-hardness due to~\cite{AA12} and~\cite{AA12,KS13},
respectively. In this paper, we settle the open question regarding the computational complexity of the
problem, by presenting a fully polynomial-time approximation scheme.

At the core of our approach is a transformation of the original preemptive problem into a non-preemptive problem, where each task is replaced by a
polynomial number of \emph{pieces}. At first sight, it may seem
counterintuitive to transform a preemptive problem into a harder
non-preemptive problem, especially as Bampis et al. \cite{BKLLN13}
show that (for the problem of speed scaling alone) the
ratio between an optimal preemptive and an optimal non-preemptive
solution on the same instance can be very high. However, this does not
apply in our case, as we consider the non-preemptive problem on a
modified instance (we seek to schedule each piece of a task
non-preemptively and not the whole task itself). Furthermore, in our analysis, we make use of a
particular lexicographic ordering, which exploits the advantages of
preemption. More specifically, we design a dynamic program that outputs a schedule
which is optimal among a restricted class of non-preemptive schedules. The definition of this class is based on a discretization of the time horizon by a careful choice of polynomially many time points. Roughly speaking, the class is comprised of these schedules that start and end the processing of each piece at such a time point, and satisfy a certain constraint regarding the processing order of the pieces. To prove that a near-optimal schedule in this class exists, we perform a series of transformations from a lexicographically
minimal optimal schedule for the original problem to a schedule of the
above class, while the energy consumption increases by at most a
factor of $(1+\epsilon)$. The lexicographic ordering is crucial to
ensure that we get the correct ordering among the pieces without
further increasing the energy consumption.  

We remark that Baptiste~\cite{B06} used a dynamic
program of similar structure for the case of unit-sized tasks and a
fixed-speed processor equipped with a sleep state. This dynamic
program is also based on a particular ordering of tasks, which, however, is not sufficient for our setting. Since we have pieces of
different sizes, the swapping argument used
in~\cite{B06} is rendered impossible. 

In Section~\ref{sec:prelim}, we describe the YDS algorithm from~\cite{YDS} for the problem
of speed scaling \emph{without} a sleep state, and then show several
properties that a schedule produced by YDS has for our problem of
speed scaling with sleep state.
We then, in Section~\ref{sec:discretized}, define a particular class of schedules that have a set
of desirable properties, and show that there exists a schedule in this class, whose energy
consumption is within a $(1+\epsilon)$-factor from optimal. 
Finally, in Section~\ref{sec:dp}, we develop an algorithm based on a dynamic
program, that outputs, in polynomial time, a schedule of minimal energy
consumption among all the schedules of the aforementioned class.

\section{Preliminaries}\label{sec:prelim}

We start by giving a short description of the YDS
algorithm presented in \cite{YDS}. For any interval $I$, let
$B(I)$ be the set of tasks whose allowed intervals are within $I$. We define
the \emph{density} of $I$ as
\begin{align*}
dens(I) = \frac{\sum_{j\in B(I)} v_j}{|I|}.
\end{align*}
Note that the average speed that any feasible schedule
uses during interval $I$ is no less than $dens(I)$. YDS works in
rounds. In the first round, the interval $I_1$ of maximal density is
identified, and all tasks in $B(I_1)$ are scheduled during $I_1$ at a
speed of $dens(I_1)$, according to the earliest deadline first policy. Then the tasks in $B(I_1)$ are removed from the
instance and the time interval $I_1$ is ``blacked out''. In general,
during round $i$, YDS identifies the interval $I_i$ of maximal density
(while disregarding blacked out times, and already scheduled jobs),
and then processes all jobs in $B(I_i)$ at a uniform speed of $dens(I_i)$.
YDS terminates when all jobs are scheduled, and its running time is
polynomial in the input size.

We remark that the speed used for the processing of jobs can never increase between two
consecutive rounds, i.e., YDS schedules the tasks by order of
non-increasing speeds. Furthermore, note that by the definition of YDS,
all the tasks scheduled in each round $i$ have their allowed interval within $I_i$.

Given any job instance $\J$, let $FAST(\J)$ be the subset of \J\ that YDS processes at a speed
greater than or equal to \scrit, and let $SLOW(\J):=\J\setminus FAST(\J)$. The following lemma is an extension of a fact proven by Irani et at.~\cite{ISG}.

\begin{lemma}
\label{lem:irani}
For any job instance $\J$, there exists an optimal schedule (w.r.t. speed scaling with sleep state) in which 
\begin{enumerate}
\item\label{irani-1} Every task in $FAST(\J)$ is processed according to YDS.
\item\label{irani-2} Every task $k \in SLOW(\J)$ is run at a uniform speed $s_k\leq \scrit$, and the processor never (actively) runs at a speed less than $s_k$ during $[r_k, d_k)$.
\end{enumerate}
We call an optimal schedule with these properties a \emph{YDS-extension} for $\J$.
\end{lemma}
\begin{proof}
To break ties among schedules with equal energy
consumption, we introduce the pseudo cost function $\int s(t)^2 dt$
(this idea was first used in~\cite{ISG}). Consider a minimal pseudo
cost schedule $Y$, so that $Y$ satisfies property \ref{irani-1}, and
minimizes the energy consumption among all schedules fulfilling this
property. It was shown in \cite{ISG} that $Y$ is optimal for instance
$\J$, and that under $Y$
\begin{equation}\tag{$\ast$}
\parbox{0.8\linewidth}{every task $k \in SLOW(\J)$ is run at a uniform speed $s_k$, and the processor never (actively) runs at a speed less than $s_k$ during those portions of $[r_k, d_k)$ where no job from FAST(\J) is processed.} 
\end{equation}
It therefore remains to prove that the speeds $s_k$ are no higher than
$\scrit$. For the sake of contradiction, assume that there exists a
job $j \in SLOW(\J)$ which is processed at speed higher than
\scrit. Let $\mathcal{I}$ be a maximal time interval, so that (i)
$\mathcal{I}$ includes at least part of the execution of $j$, and (ii) at any time point $t \in \mathcal{I}$ the processor either runs strictly faster than $\scrit$, or executes a job from $FAST(\J)$. Then there must exist a job $k \in SLOW(\J)$ (possibly $k=j$) which is executed to some extent during $\mathcal{I}$, and whose allowed interval is not contained in $\mathcal{I}$ (otherwise, when running YDS, the density of $\mathcal{I}$ after the jobs in $FAST(\J)$ have been scheduled is larger than $\scrit$, contradicting the fact that YDS processes all remaining jobs slower than $\scrit$). By the maximality of $\mathcal{I}$, there exists some interval $\mathcal{I'} \subseteq [r_k,d_k)$ right before $\mathcal{I}$ or right after $\mathcal{I}$, during which no job from $FAST(\J)$ is executed, and the processor either runs with speed at most \scrit \ or resides in the sleep state. The first case contradicts property~($\ast$), as $k$ is processed during $\mathcal{I}$ and thus at speed $s_k > \scrit$. In the second case, we can use a portion of $\mathcal{I'}$ to slightly slow down $k$ to a new speed $s'$, such that $\scrit < s' < s_k$. The resulting schedule $Y'$ has energy consumption no higher than $Y$, as $P(s)/s$ is non-decreasing for $s\geq \scrit$. Furthermore, if $\mathcal{C}_p$ is the pseudo cost of $Y$, then $Y'$ has pseudo cost $\mathcal{C}_p-v_k s_k+v_k s' < \mathcal{C}_p$. This contradicts our assumptions on $Y$. 
\end{proof}

By the previous lemma, we may use YDS to schedule the jobs in
$FAST(\J)$, and need to find a good schedule only for the remaining
jobs (which are exactly $SLOW(\J)$). To this end, we transform the input instance $\J$ to an
instance $\J'$, in which the jobs $FAST(\J)$ are replaced by
dummy jobs. This introduction of dummy tasks bears resemblance to the
approach of~\cite{AA12}. We then show in Lemma~\ref{lem:reduction}, that any schedule for $\J'$ with a certain property, can be transformed to a schedule for $\J$ without any degradation in the approximation factor.

Consider the schedule $S_{YDS}$ that algorithm YDS produces on
$\J$. Let $I_i=[y_i,z_i)$, $i=1,\dots, \ell$ be the $i$-th maximal interval in which $S_{YDS}$ continuously runs
at a speed greater than or equal to \scrit , and let $T_1,\dots, T_m$ be the remaining
intervals in $[0,d_{max})$ not covered by intervals
$I_1,I_2,\dots, I_\ell$. Furthermore, let $\mathcal{T}:=\cup_{1\le k \le m}\ T_k$. Note that the intervals
$I_i$ and $T_i$ partition the time horizon $[0,d_{max})$, and
furthermore, by the way YDS is defined, 
every job $j\in FAST(\J)$ is active in exactly one interval $I_i$,
and is not active in any interval $T_i$. The opposite does not
necessarily hold, i.e., a job $j\in SLOW(\J)$ may be active in several
(consecutive) intervals $I_i$ and $T_{i'}$. We transform $\J$ to a job instance $\J'$ as
follows:

\begin{itemize}
\item For every job $j\in SLOW(\J)$,
% \begin{itemize}
% \item $v_{j'}:=v_j$.
% \item 
if there exists an $i$ such that $r_j \in I_i$ (resp. $d_j\in I_i$),
then we set $r_{j} := z_i$ (resp. $d_{j} := y_i$), else we keep the job
as it is.
%\end{itemize}
\item For each $I_i$, we replace all tasks $j\in FAST(\J)$ that are
active in $I_i$ by a single task $j^d_i$ with release time at $y_i$, deadline
at $z_i$ and processing volume of $\max\{1, |I_i|\cdot\scrit\}$.
\end{itemize}
Clearly, the above transformation can be done in polynomial time. Note
that after the transformation, there is no release time or deadline in the interior of any interval $I_i$. Furthermore, we have the following
proposition:

\begin{proposition}
\label{prop:fast-slow}
$FAST(\J') = \{j_i^d: 1\leq i\leq \ell\}$ and $SLOW(\J') = SLOW(\J)$. 
\end{proposition}
\begin{proof}
Since $\J' = \{j_i^d: 1\leq i\leq \ell\}\cup SLOW(\J)$, and furthermore $SLOW(\J')$ and $FAST(\J')$ are disjoint sets, it
suffices to show that (i) $FAST(\J') \supseteq \{j_i^d: 1\leq i\leq \ell\}$
and that (ii) $SLOW(\J')\supseteq SLOW(\J)$. 

For (i), we observe that no task $j_i^d$ can be feasibly scheduled at a uniform speed less than $\scrit$. As YDS uses a uniform speed for each task, these jobs must belong to $FAST(\J')$.

For (ii), consider the execution of YDS on $\J'$. More specifically, consider the first round when a job from $SLOW(\J)$ is scheduled. Let $\mathcal{I}$ be the maximal density interval of this round, and let $\mathcal{J}_S$ and $\mathcal{J}_d$ be the sets of jobs from $SLOW(\J)$ and $\{j_i^d: 1\leq i\leq \ell\}$, respectively, that are scheduled in this round (note that $\mathcal{I}$ contains the allowed intervals of these jobs). As the speed used by YDS is non-increasing from round to round, it suffices to show that $dens(\mathcal{I})<\scrit$. 

Consider a partition of $\mathcal{I}$ into maximal intervals $\Lambda_1, \dots, \Lambda_a$, s.t. each $\Lambda_k$ is contained in some interval $I_i$ or $T_i$. Then
\begin{multline*}
dens(\mathcal{I}) = \frac{\sum_{j \in \mathcal{J}_d}v_j}{|\mathcal{I}|} + \frac{\sum_{j \in \mathcal{J}_S} v_j}{|\mathcal{I}|} = 
\sum\limits_{\Lambda_k \not\subseteq
  \mathcal{T}}\left(\frac{|\Lambda_k|}{|\mathcal{I}|}
  dens(\Lambda_k)\right) + \frac{\sum_{\Lambda_k \subseteq
    \mathcal{T}} |\Lambda_k|}{|\mathcal{I}|} \cdot \frac{\sum_{j \in
    \mathcal{J}_S} v_j}{\sum_{\Lambda_k\subseteq \mathcal{T}}
  |\Lambda_k|}\\
\le \Big(\sum\limits_{\Lambda_k \not\subseteq
  \mathcal{T}} \frac{|\Lambda_k|}{|\mathcal{I}|}\Big)
  dens(\mathcal{I})
+ \Big(1-\sum\limits_{\Lambda_k \not\subseteq
  \mathcal{T}} \frac{|\Lambda_k|}{|\mathcal{I}|}\Big) \cdot \frac{\sum_{j \in
    \mathcal{J}_S} v_j}{\sum_{\Lambda_k\subseteq \mathcal{T}}
  |\Lambda_k|},
\end{multline*}
since no $\Lambda_k$ can have a density larger than
$dens(\mathcal{I})$ (because $\mathcal{I}$ is the interval of maximal density). It follows that
\[dens(\mathcal{I}) \le \frac{\sum_{j \in \mathcal{J}_S} v_j}{\sum_{\Lambda_k\subseteq \mathcal{T}} |\Lambda_k|}\ .\]
Furthermore, by the definition of $SLOW(\J)$, it is possible to schedule all tasks in $\mathcal{J}_S$ during $\mathcal{I} \cap \mathcal{T}$, at a speed slower than $\scrit$ (since none of the steps in the transformation from $\J$ to $\J'$ reduces the time any task is active during $\mathcal{T}$). Together with the previous inequality, this implies $dens(\mathcal{I}) < \scrit$.
\end{proof}

The following lemma suggests that for obtaining an FPTAS for instance
$\J$, it suffices to give an FPTAS for instance $\J'$, as long as we schedule the tasks $j_i^d$ exactly in their allowed intervals $I_i$.

\begin{lemma}
\label{lem:reduction}
Let $S'$ be a schedule for input instance $\J'$, that (i) processes
each task $j_i^d$ exactly in its allowed interval $I_i$ (i.e. from $y_i$ to $z_i$), and (ii) is
a $c$-approximation for $\J'$. Then $S'$ can be transformed in polynomial time into a schedule $S$
that is a $c$-approximation for input instance $\J$.
\end{lemma}
\begin{proof}
Given such a schedule $S'$, we leave the processing in the intervals
$T_1,\dots, T_m$ unchanged, and replace for each interval $I_i$ the
processing of job $j_i^d$ by the original YDS-schedule $S_{YDS}$ during $I_i$. It is easy to see that the
resulting schedule $S$ is a feasible schedule for $\J$. We now argue about the approximation factor.

Let OPT be a YDS-extension for $\J$, and let OPT$'$ be a YDS-extension for $\J'$. Recall that $E(\cdot)$ denotes the energy consumption of a schedule (including wake-up costs). Additionally, let $E^I(S)$ denote the total energy consumption of $S$ in all
intervals $I_1,\dots, I_\ell$ without wake-up costs (i.e. the energy consumption for processing or being active but idle during those intervals), and
define similarly $E^I(S')$, $E^I(OPT)$, and $E^I(OPT')$ for the schedules $S'$, OPT,
and OPT$'$, respectively. Since $S'$ is a $c$-approximation for $\J'$, we have
\begin{align*}
E(S') \le c E(OPT').
\end{align*}
Note that OPT$'$ schedules exactly the task $j_i^d$ in each $I_i$ (using the entire interval for it) by Proposition~\ref{prop:fast-slow}, and thus each of the schedules $S$, $S'$, OPT, and OPT$'$ keeps the processor active during every entire interval $I_i$. Therefore
\begin{align*}
E(S) - E(S') = E^I(S) - E^I(S'), 
\end{align*}
since $S$ and $S'$ have the same wake-up costs and do not differ in the intervals $T_1,\dots, T_m$. Moreover,
\begin{align*}
E(OPT) - E(OPT') = E^I(OPT) - E^I(OPT'),
\end{align*}
as $E(OPT) - E^I(OPT)$ and $E(OPT') - E^I(OPT')$ are both equal to the optimal energy consumption during $\mathcal{T}$, of any schedule that processes the jobs $SLOW(\J)$ in $\mathcal{T}$ and resides in the active state during each interval $I_i$ (including all wake-up costs of the schedule). 
Clearly, $E^I(S) = E^I(OPT)$, and since both $S'$ and OPT$'$ schedule
exactly the task $j_i^d$ in each $I_i$ (using the entire interval for it), we have that $E^I(S') \geq E^I(OPT')$. Therefore
\begin{align*}
E(S) - E(S') \le E(OPT) - E(OPT').
\end{align*}
We next show that $0 \le E^I(OPT) - E^I(OPT') = E(OPT) - E(OPT')$, which implies
\begin{align*}
\label{eq:toshow}
E(S) &\le E(OPT) - E(OPT') + E(S') \le c (E(OPT) - E(OPT')) + E(S') \\ &\le c (E(OPT) - E(OPT')) + c E(OPT') \le c E(OPT).
\end{align*}
Since YDS (when applied to $\J$) uses a speed of at least \scrit \ for each interval $I_i$, and furthermore processes a volume of at least $v_{min}=1$ in each such interval, OPT runs with an average speed of at least $\max\{1/ |I_i|,\scrit\}$ in each $I_i$. On the other hand, OPT$'$ runs with speed exactly $\max\{1/ |I_i|,\scrit\}$ during $I_i$, and therefore $E^I(OPT) \ge E^I(OPT')$.
\end{proof}

\section{Discretizing the Problem}
\label{sec:discretized}

After the transformation in the previous section, we have an instance $\J'$. In this section, 
we show that there exists a ``discretized'' schedule for $\J'$, whose energy consumption 
is at most $1+\epsilon$ times that of an optimal schedule for $\J'$. In the next section, 
we will show how such a discretized schedule can be found by dynamic programming. 

Before presenting formal definitions and technical details, we here first sketch 
the ideas behind our approach. 

A major challenge of the original problem is that we need to deal with an infinite number of possible schedules. We overcome this intractability by ``discretizing'' the problem as follows: (1) we break each job in $SLOW(\J')$ into smaller pieces, and 
(2) we create a set of time points 
and introduce the additional constraint that each piece of a job has to start and end at these time points. 
The number of the introduced time points and job pieces are both
polynomial, which greatly limits 
the amount of guesswork we have to make in the dynamic program. The
challenge is to both find such a discretization and argue that it
does not increase the optimal energy consumption by too much. 

\subsection{Further Definitions and Notation}

We first define the set $W$ of 
time points. 
%These time points will be useful to define a class of
%schedules, called \emph{well-ordered discretized schedules}. 
%We will eventually show
%that there exists a well-ordered discretized schedule with energy consumption not too far from
%optimal, and also provide an algorithm to find an optimal well-ordered
%discretized schedule.
Given an error parameter $\epsilon > 0$, let $\delta := \min\{\frac14, \frac{\epsilon}{4}
\frac{P(\scrit )}{P(2\scrit )-P(\scrit )}\}$. Furthermore, let
\begin{align*}
W' := \bigcup_{j \in \J'} \{r_j,d_j\}.
\end{align*}
Consider the elements of $W'$ in sorted order, and let $t_i, 1\le i
\le |W'|$ be $i$-th element of $W'$ in this order. We call an
interval $[t_i, t_{i+1})$ for $1\le i \le |W'|-1$ a
\emph{zone}, and observe that every zone is either equal to some interval $I_i$ or contained in $\mathcal{T}$.
%Furthermore, let $i\in T''$ if and only if there exists no $j$ so that
%$[t_i,t_{i+1}) = I_j$.

% Then we have a set
%of $|T'|-1$ \emph{zones}, where zone $z_i, 1\le i\le |T'|-1$ is
%defined as the time interval $[t_i,t_{i+1})$.
%
%The following procedure defines the set of time points $T$:
%For each zone $z_i$ for $1\le i \le |T'|-1$, add time point
%$t_i$ into $T$, and let $q$ be the largest $j$ so that
%$(1+\delta)^j\frac{v_{min}\delta}{4n^2\scrit}\le t_{i+1}-t_i$. For
%each combination of $j$ from $1$ to $q$, and $r$ from $1$ to $16n^6\lceil 1/\delta
%\rceil^2 (1+ \lceil 1/\delta \rceil )$, add the time points

For each $i$ in $1,\dots , |W'|-1$, let $x(i)$ be the largest integer $j$ so that 
\begin{align*}
(1+\delta)^j\frac{1}{4n^2\scrit (1+\delta) \lceil 1/\delta \rceil}\le t_{i+1}-t_i.
\end{align*}
We are now ready to define the set of time points $W$ as follows:

\begin{multline*}
W := W' \cup \bigg(\ \underset{1\le r \le 16n^6 \lceil 1/\delta\rceil ^2(1+\lceil
1/\delta \rceil )}{\underset{0\le j \le
x(i)}{\underset{i \text{ s.t. }[t_i,t_{i+1})\subseteq \mathcal{T}}{\bigcup}}}
\left \{ t_i + r\cdot \frac{(1+\delta)^j\frac{1}{4n^2\scrit (1+\delta) \lceil 1/\delta \rceil
}}{16n^6\lceil 1/\delta \rceil^2 (1+\lceil 1/\delta \rceil )}, \right.\\
\left. t_{i+1} - r\cdot \frac{(1+\delta)^j\frac{1}{4n^2\scrit (1+\delta) \lceil 1/\delta \rceil
}}{16n^6\lceil 1/\delta \rceil^2 (1+\lceil 1/\delta \rceil )} \right\}\bigg ).
\end{multline*}

Let us explain how these time points in $W$ come about. 
As we will show later (Lemma~\ref{lem:optproperties}(\ref{prop:block})), there exists a certain optimal schedule for $\J'$ 
in which each zone $[t_i,t_{i+1})\subseteq \mathcal{T}$ contains at most one contiguous maximal processing interval, 
and this interval ``touches'' either $t_i$ or $t_{i+1}$ (or both). The geometric series 
$(1+\delta)^j\frac{1}{4n^2\scrit (1+\delta) \lceil 1/\delta \rceil}$ of time points are used to 
approximate the ending/starting time of this maximal processing interval. For each guess of the ending/starting 
time, we 
%Intuitively, we assume that during a zone $[t_i, t_{i+1})\subseteq \mathcal{T}$, jobs are
%processed for a time of $(1+\delta)^j\frac{1}{4n^2\scrit (1+\delta) \lceil 1/\delta \rceil}$
%for some $j$, and this interval either starts at $t_i$ or ends at
%$t_{i+1}$. We then 
split the guessed interval, during which the job pieces (to be defined formally immediately) are to be 
processed, into $16n^6\lceil 1/\delta\rceil^2 (1+\lceil
1/\delta\rceil )$ many intervals of equal length.
An example of the set $W$ for a given zone can be seen in
Figure~\ref{fig:W}.

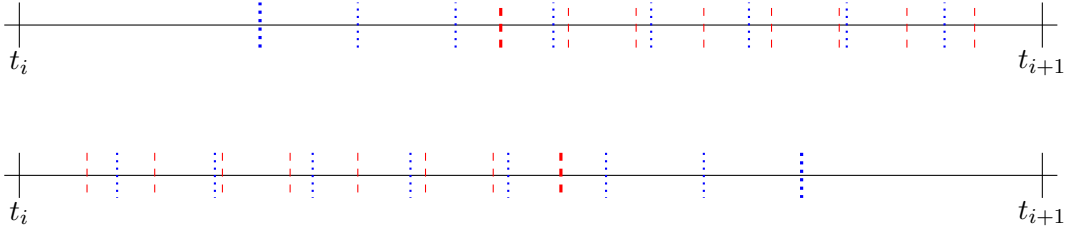
\begin{figure}
\centering 
\begin{tikzpicture}
\begin{scope}
\draw (0,0) -- (14,0);
\draw (0.2,-0.3) -- (0.2,0.3);
\node () at (0.2,-0.5) {$t_i$};
\draw (13.8,-0.3) -- (13.8,0.3);
\node () at (13.8,-0.5) {$t_{i+1}$};
\draw[very thick,red, dashed] (7.4,0.3) -- (7.4,-0.3);
\draw[dashed,red] (1.1,0.3) -- (1.1,-0.3);
\draw[dashed,red] (2,0.3) -- (2,-0.3);
\draw[dashed,red] (2.9,0.3) -- (2.9,-0.3);
\draw[dashed,red] (3.8,0.3) -- (3.8,-0.3); 
\draw[dashed,red] (4.7,0.3) -- (4.7,-0.3);
\draw[dashed,red] (5.6,0.3) -- (5.6,-0.3);
\draw[dashed,red] (6.5,0.3) -- (6.5,-0.3);
\draw[very thick,blue, dotted] (10.6,0.3) -- (10.6,-0.3);
\draw[dotted,blue,thick] (1.5,0.3) -- (1.5,-0.3);
\draw[dotted,blue,thick] (2.8,0.3) -- (2.8,-0.3);
\draw[dotted,blue,thick] (4.1,0.3) -- (4.1,-0.3);
\draw[dotted,blue,thick] (5.4,0.3) -- (5.4,-0.3); 
\draw[dotted,blue,thick] (6.7,0.3) -- (6.7,-0.3);
\draw[dotted,blue,thick] (8,0.3) -- (8,-0.3);
\draw[dotted,blue,thick] (9.3,0.3) -- (9.3,-0.3);
\end{scope}
\begin{scope}[rotate around={180:(7,0)}]%[cm = {1,1,0,0,(0,0)}]
\draw (0,-2) -- (14,-2);
\draw (0.2,-2.3) -- (0.2,-1.7);
\node () at (0.2,-1.5) {$t_{i+1}$};
\draw (13.8,-2.3) -- (13.8,-1.7);
\node () at (13.8,-1.5) {$t_{i}$};
\draw[very thick,red, dashed] (7.4,-1.7) -- (7.4,-2.3);
\draw[dashed,red] (1.1,-1.7) -- (1.1,-2.3);
\draw[dashed,red] (2,-1.7) -- (2,-2.3);
\draw[dashed,red] (2.9,-1.7) -- (2.9,-2.3);
\draw[dashed,red] (3.8,-1.7) -- (3.8,-2.3); 
\draw[dashed,red] (4.7,-1.7) -- (4.7,-2.3);
\draw[dashed,red] (5.6,-1.7) -- (5.6,-2.3);
\draw[dashed,red] (6.5,-1.7) -- (6.5,-2.3);
\draw[very thick,blue, dotted] (10.6,-1.7) -- (10.6,-2.3);
\draw[dotted,blue,thick] (1.5,-1.7) -- (1.5,-2.3);
\draw[dotted,blue,thick] (2.8,-1.7) -- (2.8,-2.3);
\draw[dotted,blue,thick] (4.1,-1.7) -- (4.1,-2.3);
\draw[dotted,blue,thick] (5.4,-1.7) -- (5.4,-2.3); 
\draw[dotted,blue,thick] (6.7,-1.7) -- (6.7,-2.3);
\draw[dotted,blue,thick] (8,-1.7) -- (8,-2.3);
\draw[dotted,blue,thick] (9.3,-1.7) -- (9.3,-2.3);
\end{scope}
\end{tikzpicture}
\caption{We assume that $r=1\dots 8$ and that $x(i)=2$. The red dashed
points correspond to $j=1$ and the blue dotted points to $j=2$. For
clarity, we drew the points defined from $t_i$ and
from $t_{i+1}$ in two separate pictures. Note that for each $j$ the number of
points is the same and the points of the same color are at equal distance from each other.}
\label{fig:W}
\end{figure}

Note that $|W|$ is polynomial in the input size and $1/\epsilon$.

\begin{definition}
We split each task $j\in SLOW(\J')$ into $4n^2\lceil1/\delta\rceil$
equal sized \emph{pieces}, and also consider each job $j_i^d$ as a single piece
on its own. For every piece~$u$ of some job $j$, let $job(u) := j$,
$r_u:=r_j$, $d_u:=d_j$, and $v_u:=v_j/(4n^2\lceil 1/\delta\rceil)$ if $j\in SLOW(\J')$, and $v_u:=v_j$ otherwise. Furthermore, let $D$ denote the set of all pieces from all jobs. 
\end{definition}

Note that $|D|= \ell + |SLOW(\J')| \cdot 4n^2\lceil1/\delta\rceil$ is polynomial in the input size and $1/\epsilon$. We now define an ordering of the pieces in $D$. 

\begin{definition}
Fix an arbitrary ordering of the jobs in $\J'$ , s.t. for any two different jobs $j$ and $j'$, $j\prec j'$ implies $r_j \le r_{j'}$. Now extend this ordering to the set of pieces, s.t. for any two pieces $u$ and $u'$, there holds
\begin{align*}
u\prec u' \Rightarrow job(u) \preceq job(u').
\end{align*}
\end{definition}

We point out that any schedule for $\J'$ can also be seen as a schedule for $D$, by implicitly assuming that the pieces of any fixed job are processed in the above order. 

We are now ready to define the class of discretized schedules.

\begin{definition}\label{def:discr}
A \emph{discretized schedule} is a schedule for $\J'$ that satisfies the following two
properties:
\begin{enumerate}[(i)]
\item\label{prop:disc1} Every piece is completely processed in a
single zone, and without preemption.
%For any $i$ in $\{1, \ldots, |W'|-1\}$, there
%for e holds that exactly an integer number of pieces is executed in $[t_i,t_{i+1})$.
\item\label{prop:disc2} The execution of every piece starts and ends at a
time point from the set $W$.
\end{enumerate}
A discretized schedule $S$ is called \emph{well-ordered} if and only if
\begin{enumerate}[(i)]\setcounter{enumi}{2}
\item\label{prop:disc3} For any time point $t$, such that in $S$ a piece $u$ ends at $t$, $S$ schedules all pieces $u'\succ u$ with $d_{u'}\geq t$ after $t$.
\end{enumerate}
\end{definition}

Finally, we define a particular ordering over possible schedules, which
will be useful in our analysis.

\begin{definition}
Consider a given schedule. For every job $j \in \J'$, and every $x \le v_j$, let
$c_j(x)$ denote the earliest time point at which volume $x$ of job $j$
has been finished under this schedule. Furthermore, for any $j \in \J'$, we define
\begin{align*}
q_j:=\int_{0}^{v_j} c_j(x)dx.
\end{align*}

Let $j_1\prec j_2 \prec \dots \prec j_{|\J'|}$ be the jobs in $\J'$. A schedule $S$ is \emph{lexicographically smaller} than a schedule $S'$ if and
only if it is lexicographically smaller with respect to the vector $(q_{j_1}, q_{j_2}, \dots, q_{j_{|\J'|}})$.
\end{definition}

Observe that shifting the processing interval of any fraction of some job
$j$ to an earlier time point (without affecting the other processing
times of $j$) decreases the value of $q_j$.

\subsection{Existence of a Near-Optimal Discretized Schedule}

In this section, we show that there exists a YDS-extension for $\J'$ with certain nice properties (recall that 
a YDS-extension is an optimal schedule satisfying the properties of Lemma \ref{lem:irani}),  
that such a YDS-extension can be transformed into a well-ordered
discretized schedule, and that the speed 
of the latter, at all time, is at most $(1+\delta)^3$ times that of the former. This fact essentially 
guarantees the existence of a well-ordered discretized schedule with energy consumption at most $1+\epsilon$ that of an
optimal schedule for $\J'$.

\begin{lemma}\label{lem:optproperties}
Let OPT be a lexicographically minimal YDS-extension for $\J'$. Then
the following hold:
\begin{enumerate}
\item\label{prop:DummiesInIs} Every task $j_i^d$ is scheduled exactly in its allowed interval $I_i$.
\item\label{prop:block} Every zone $[t_i,t_{i+1})\subseteq \mathcal{T}$ has the following two properties:
\begin{enumerate}[(a)]
\item There is at most one contiguous maximal processing interval
  within $[t_i,t_{i+1})$, and this interval either starts at $t_i$
  and/or ends at $t_{i+1}$. We call this interval the \emph{block} of zone $[t_i,t_{i+1})$.
\item OPT uses a uniform speed of at most $\scrit$ during this block.
\end{enumerate}
\item\label{prop:OPTorder} There exist no two jobs $j' \succ j$,
  such that  a portion of $j$ is processed after some portion of $j'$, and before $d_{j'}$.
\end{enumerate}
\end{lemma}

\begin{proof}
\begin{enumerate}
\item Since $FAST(\J') = \{j_i^d: 1\leq i\leq \ell\}$ (by
  Proposition~\ref{prop:fast-slow}), and OPT is a
YDS-extension, it follows that each $j_i^d$ is processed exactly
in its allowed interval $I_i$.
\item 
\begin{enumerate}[(a)]
\item Assume for the sake of
contradiction that $[t_i,t_{i+1})\subseteq \mathcal{T}$ contains a number of maximal intervals $N_1,N_2,\dots, N_\psi$ (ordered from left to
right\footnote{For any two time points $t_1<t_2$, we 
say that $t_1$ is to the \emph{left} of $t_2$, and $t_2$
is to the \emph{right} of $t_1$.}) during which jobs are being processed, with $\psi \ge 2$. Let
$M_1,M_2,\dots, M_{\psi'}$ (again ordered from left to
right) be the remaining maximal intervals in
$[t_i,t_{i+1})$, so that $N_1,\dots, N_\psi$ and $M_1,\dots, M_{\psi'}$ partition the zone
$[t_i,t_{i+1})$. Furthermore, note that for each $i=1,\dots, \psi'$, the
processor is either active but idle or asleep during the whole interval
$M_i$, since otherwise setting the processor asleep during the whole
interval $M_i$  would incur a strictly smaller energy consumption.

%For each pair of $N_i$, $N_{i+1}$, $i=1,\dots, \psi-1$, such that
%the interval between $N_i$ and $N_{i+1}$ is $M_j$, we shift
%$N_{i+1}$ to the left by $|M_j|$ time units. This merges $M_j$
%and $M_{j+1}$ into a new interval during which no task is being
%processed (if $j=\psi'$ we assume that $M_{j+1}$ has length $0$). If the processor was active during both $M_j$ and
%$M_{j+1}$ then we keep the processor active during this new
%interval, else we transition it to the sleep state. We observe
%that the resulting schedule still is a YDS-extension, that has
%at most the energy consumption of the initial schedule, but is
%lexicographically smaller.

We modify the schedule by shifting the intervals $N_i$, $i = 2, \dots, \psi$ to the left, so that $N_1, N_2, \dots, N_\psi$ now form a single contiguous processing interval. The intervals $M_k$ lying to the right of $N_1$ are moved further right and merge into a single (longer) interval $M'$ during which no tasks are being processed.
If the processor was active during each of these intervals $M_k$, then we keep the processor active during the new interval $M'$, else we transition it to the sleep state. We observe that the resulting schedule is still a YDS-extension (note that its energy consumption is at most that of the initial schedule), but is lexicographically smaller.

For the second part of the statement, assume that there exists
exactly one contiguous maximal processing interval $N_1$ within
$[t_i,t_{i+1})$, and that there exist two
$M$-intervals, $M_1$ and $M_2$ before and after $N_1$, respectively.

We consider two cases:
\begin{itemize}
\item The processor is active just before $t_i$, or the processor is
asleep both just before $t_i$ and just after $t_{i+1}$: In this case
we can shift $N_1$ left by $|M_1|$ time units, so that it starts
at $t_i$. Again, we keep the processor active during
$[t_i+|N_1|, t_{i+1})$ only if it was active during both $M_1$ and
$M_2$. As before, the resulting schedule remains a
YDS-extension, and is
lexicographically smaller.
\item The processor is in the sleep state just before $t_i$ but active
just after $t_{i+1}$: In this case we shift $N_1$ by $|M_2|$ time units
to the right, so that its right
endpoint becomes $t_{i+1}$. During the new idle interval $[t_i,t_i+|M_1|+|M_2|)$ we set the processor asleep. Note that in this case the processor was
asleep during $M_1$. The schedule remains a YDS-extension, but its
energy consumption
becomes strictly smaller: (i) either the processor was asleep during
$M_2$, in which case the resulting schedule uses the same energy
while the processor is active but has one wake-up operation less, or
(ii) the processor was active and idle during $M_2$, in which
case the resulting schedule saves the idle energy that was expended
during $M_2$.
\end{itemize}
\item The statement follows directly from the second property of Lemma~\ref{lem:irani} and the fact that all tasks processed during $[t_i,t_{i+1})$ belong to $SLOW(\J')$ and are active in the entire zone.
\end{enumerate}
\item Assume for the sake of contradiction that there exist two jobs
$j' \succ j$, such that a portion of $j$ is processed during an interval
$Z = [\zeta_1,\zeta_2)$, $\zeta_2 \le d_{j'}$, and some portion of $j'$ is processed during an interval
$Z' = [\zeta_1',\zeta_2')$, with $\zeta_2'\le \zeta_1$. We first
observe that both jobs belong to $SLOW(\J')$. This follows from the
fact that both jobs are active during the whole interval
$[\zeta_1',\zeta_2)$, and processed during parts of this interval,
whereas any job $j_i^d$ (which are the only jobs in $FAST(\J')$) is
processed exactly in its entire interval $[y_i,z_i)$ (by
statement~\ref{prop:DummiesInIs} of the lemma).

By the second property of Lemma~\ref{lem:irani}, both $j$ and $j'$ are
processed at the same speed. We can now apply a swap
argument. Let $L:=\min\{|Z|,|Z'|\}$. Note that OPT schedules only $j'$
during $[\zeta_2'-L,\zeta_2')$ and only $j$ during
$[\zeta_2-L,\zeta_2)$.  Swap the part of the schedule OPT
in $[\zeta_2'-L,\zeta_2')$ with the schedule in the interval
$[\zeta_2-L, \zeta_2)$. Given the above observations, it can be easily verified that the resulting schedule (i) is feasible and remains a YDS-extension, and (ii) is lexicographically smaller than OPT. 
\end{enumerate} 
\end{proof}

The next lemma shows how to transform the lexicographically minimal
YDS-extension for $\J'$ of the previous lemma 
into a well-ordered discretized schedule. This is the most crucial part of our approach. 
Roughly speaking, the transformation needs to guarantee that (1) in each zone, 
the volume of a job $j \in SLOW(\J')$ processed is an integer multiple of $v_j/(4n^2\lceil 1/\delta \rceil)$ 
(this is tantamount to making sure that each zone has integral job pieces to deal with), (2) the job pieces start 
and end at the time points in $W$, and (3) all the job pieces are processed in the ``right order''. As we will show, 
the new schedule may run at a higher speed than the given lexicographically minimal YDS-extension, but not by too much.

\begin{lemma}\label{lem:discretized}
Let OPT be a lexicographically minimal YDS-extension for $\J'$, and let $s_{\mathcal{S}}(t)$ denote the speed of schedule $\mathcal{S}$ at time $t$, for any $\mathcal{S}$ and $t$.
Then there exists a well-ordered discretized schedule $F$, such that at any
time point $t\in \mathcal{T}$, there holds
\begin{align*}
s_{F}(t)\le (1+\delta)^3s_{OPT}(t),
\end{align*}
and for every $t\notin \mathcal{T}$, there holds 
\begin{align*}
s_{F}(t)=s_{OPT}(t).
\end{align*}
\end{lemma}
\begin{proof}
Through a series of three transformations, we will transform OPT to a well-ordered
discretized schedule $F$, while upper bounding the increase
in speed caused by each of these transformations. More specifically, we will transform OPT to a schedule $F_1$ satisfying (\ref{prop:disc1}) and (\ref{prop:disc3}) of Definition \ref{def:discr}, then $F_1$ to
$F_2$ where we slightly adapt the block lengths, and finally $F_2$ to $F$ which satisfies all three properties of Definition \ref{def:discr}. Each of these transformations can
increase the speed by at most a factor $(1+\delta)$ for any $t\in \mathcal{T}$
and does not affect the speed in any interval $I_i$.

\textbf{Transformation 1 ($\mathrm{OPT} \rightarrow F_1$):} We will transform
the schedule so that 
\begin{enumerate}[(i)]
\item For each task $j\in SLOW(\J')$, an integer
multiple of $v_j/(4n^2\lceil 1/\delta \rceil)$ volume of job $j$ is
processed in each zone, and the processing order of jobs within each zone is determined by $\prec$. Together with property~\ref{prop:DummiesInIs} of Lemma~\ref{lem:optproperties}, this implies that $F_1$ (considered as a schedule for pieces) satisfies Definition~\ref{def:discr}(\ref{prop:disc1}).
\item The well-ordered property of Definition~\ref{def:discr} is satisfied.
\item For all $t\in \mathcal{T}$ it holds that $s_{F_1}(t)\le (1+\delta)s_{OPT}(t)$, and for every
$t\notin \mathcal{T}$ it holds that $s_{F_1}(t)=s_{OPT}(t)$.
\end{enumerate}

Note that by Lemma~\ref{lem:optproperties}, every zone is either empty, filled exactly by a job $j_i^d$, or contains a single block. For any task $j\in SLOW(\J')$, and every zone $[t_i,t_{i+1})$, let
$V^i_j$ be the processing volume of task $j$ that OPT schedules in
zone $[t_i,t_{i+1})$. Since there can be at most $2n$
different zones, for every task $j$ there exists some index $h(j)$, such that $V^{h(j)}_j\ge
v_j/(2n)$.

For every task $j\in SLOW(\J')$, and every $i\neq h(j)$, we
reduce the load of task $j$ processed in $[t_i,t_{i+1})$, by setting it to 
\begin{align*}
\mathcal{V}_j^i = \Big \lfloor V_j^i / \frac{v_j}{4n^2\lceil 1/\delta \rceil)}\Big \rfloor \cdot \frac{v_j}{4n^2\lceil 1/\delta \rceil}.
\end{align*}
Finally, we set the volume of $j$
processed in $[t_{h(j)},t_{h(j)+1})$ to $\mathcal{V}_j^{h(j)} = v_j - \sum_{i\neq h(j)}
\mathcal{V}_j^i$. To keep the schedule feasible, we process the new volume of each non-empty zone $[t_i,t_{i+1})\subseteq \mathcal{T}$ in the zone's original block $B_i$, at a uniform speed of $\sum_{j\in SLOW(\J')} (\mathcal{V}_j^i)/|B_i|$. Here, the processing order of the jobs within the block is determined by $\prec$.

Note that in the resulting schedule $F_1$, a job may be processed at
different speeds in different zones, but each zone uses only one
constant speed level.

It is easy to see that $F_1$ is a feasible schedule in which for each task $j\in SLOW(\J')$, an integer multiple of $v_j/(4n^2\lceil
1/\delta \rceil)$ volume of $j$ is processed in each zone, and that $\mathcal{V}_j^i \le V_j^i$ for all $i\ne h(j)$. Furthermore, if $i = h(j)$, we have that $\mathcal{V}_j^i
- V_j^i \le v_j/(2n\lceil 1/\delta \rceil)$, and $V_j^i\ge
v_j/(2n)$. It follows that $\mathcal{V}_j^i\le V_j^i+V_j^i/\lceil 1/\delta \rceil \le(1+\delta)V_j^i$ in this case, and
therefore $s_{F_1}(t) \le
(1+\delta)s_{OPT}(t)$ for all $t\in \mathcal{T}$. We note here, that for every task $j_i^d$, and
the corresponding interval $I_i$, nothing changes during the
transformation. 

We finally show that $F_1$ satisfies the well-ordered property of
Definition~\ref{def:discr}. Assume for the sake of contradiction that
there exists a piece $u$ ending at some $t$, and there exists a piece
$u'\succ u$ with $d_{u'}\ge t$ that is scheduled before $t$. Recall
that we can implicitly assume that the pieces of any fixed job are
processed in the corresponding order $\prec$. Therefore $job(u') \succ job(u)$, by
definition of the ordering $\prec$ among pieces. Furthermore, if $[t_k,t_{k+1})$ and $[t_{k'},t_{k'+1})$ are the
zones in which $u$ and $u'$, respectively, are scheduled, then $k'<k$,
as $k' = k$ would contradict $F_1$'s processing order of jobs inside a zone. Also note that $d_{u'}
\ge t_{k+1}$, since $t\in(t_k,t_{k+1}]$, and $(t_k,t_{k+1})$ does not
contain any deadline. This contradicts property \ref{prop:OPTorder} of Lemma~\ref{lem:optproperties}, as the original schedule OPT must have processed some volume of $job(u')$ in $[t_{k'},t_{k'+1})$, and some volume of $job(u)$ in $[t_{k},t_{k+1})$. 

\textbf{Transformation 2 ($F_1\rightarrow F_2$):} In this transformation, we slightly modify the block lengths, as a preparation for Transformation 3. For every non-empty zone $[t_i,t_{i+1})\subseteq \mathcal{T}$, we increase the uniform speed of its block until it has a length of $(1+\delta)^j\frac{1}{4n^2\scrit (1+\delta) \lceil 1/\delta \rceil}$ for some integer $j\ge 0$, keeping one of its endpoints fixed at $t_i$ or $t_{i+1}$. Note that in $F_1$, the block had length at least $\frac{1}{4n^2\scrit (1+\delta) \lceil 1/\delta \rceil}$, since it contained a volume of at least $1/(4n^2\lceil 1/\delta \rceil)$, and the speed in this zone was at most $(1+\delta)\scrit$. The speedup needed for this modification is clearly at most $(1+\delta)$.

As this transformation does not change the processing order of any pieces nor the zone in which any piece is scheduled, it preserves the well-ordered property of Definition~\ref{def:discr}.

\textbf{Transformation 3 ($F_2\rightarrow F$):} In this final
transformation, we want to establish
Definition~\ref{def:discr}(\ref{prop:disc2}). To this end, we shift
and compress certain pieces in $F_2$, such that every execution interval starts and ends at a time point from $W$ (this is already true for pieces corresponding to tasks $j_i^d$). The procedure resembles a transformation done in~\cite{HuangOtt}. For any zone $[t_i,t_{i+1})\subseteq \mathcal{T}$, we do the following:
Consider the pieces that $F_2$ processes within the zone $[t_i, t_{i+1})$, and denote this set of pieces by $D_i$. If $D_i = \emptyset$, nothing needs to be done. Otherwise, let $\beta$ be the integer such that $(1+\delta)^\beta\frac{1}{4n^2\scrit (1+\delta) \lceil 1/\delta \rceil}$ is the length of the block in this zone, and let 
\begin{align*}
\Delta := \frac{(1+\delta)^\beta\frac{1}{4n^2\scrit (1+\delta) \lceil 1/\delta \rceil}}{16n^6\lceil 1/\delta \rceil^2 (1+\lceil 1/\delta \rceil )}. 
\end{align*}
Note that in the definition of $W$, we introduced $16n^6\lceil 1/\delta \rceil^2(1+\lceil 1/\delta\rceil)$ many time points (for $j=\beta$ and $r=1,\dots, 16n^6\lceil 1/\delta \rceil^2(1+\lceil 1/\delta\rceil)$) that subdivide this block into $16n^6\lceil 1/\delta \rceil^2(1+\lceil 1/\delta\rceil)$ intervals of length $\Delta$. Furthermore, since $|D_i| \leq 4n^3\lceil1/\delta\rceil$, there must exist a piece $u \in D_i$ with execution time $\Gamma_u \geq 4n^3\lceil 1/\delta \rceil (1+\lceil 1/\delta \rceil ) \Delta$. 
We now partition the pieces in $D_i \setminus u$ into $D^+$, the
pieces processed after $u$, and $D^-$, the pieces processed before
$u$. First, we restrict our attention to $D^+$. Let $q_1,\ldots,
q_{|D^+|}$ denote the pieces in $D^+$ in the order they are processed
by $F_2$. Starting with the last piece $q_{|D^+|}$, and going down to
$q_1$, we modify the schedule as follows. We keep the end of
$q_{|D^+|}$'s execution interval fixed, and shift its start to the
next earlier time point in $W$, reducing its uniform execution speed accordingly. At the same time, to not produce any overlappings, we shift the execution intervals of all $q_k,\ k<|D^+|$ by the same amount to the left (leaving their lengths unchanged). Eventually, we also move the execution end point of $u$ by the same amount to the left (leaving its start point fixed). This shortens the execution interval of $u$ and ``absorbs'' the shifting of the pieces in $D^+$ (note that the processing speed of $u$ increases as its interval gets shorter). We then proceed with $q_{|D^+|-1}$, keeping its end (which now already resides at a time point in $W$) fixed, and moving its start to the next earlier time point in $W$. Again, the shift propagates to earlier pieces in $D^+$, which are moved by the same amount, and shortens $u$'s execution interval once more. When all pieces in $D^+$ have been modified in this way, we turn to $D^-$ and apply the same procedure there. This time, we keep the start times fixed and instead shift the right end points of the execution intervals further to the right. As before, $u$ ``absorbs'' the propagated shifts, as we increase its start time accordingly. After this modification, the execution intervals of all pieces in $D_i$ start and end at time points in $W$. 

To complete the proof, we need to argue that the speedup of piece $u$
is bounded by a factor  $(1+\delta)$. Since $|D_i| \leq 4n^3\lceil1/\delta\rceil$, $u$'s execution interval can be shortened at most $4n^3\lceil1/\delta\rceil$ times, each time by a length of at most $\Delta$. Furthermore, recall that the execution time of $u$ was $\Gamma_u \geq 4n^3\lceil 1/\delta \rceil (1+\lceil 1/\delta \rceil ) \Delta$. Therefore, its new execution time is at least $\Gamma_u- 4n^3\lceil1/\delta\rceil \Delta \geq \Gamma_u - \frac{\Gamma_u}{1+\lceil 1/\delta \rceil}$, and the speedup factor thus at most
\[\frac{\Gamma_u}{\Gamma_u-\frac{\Gamma_u}{1+\lceil 1/\delta \rceil}} = \frac{1}{1-\frac{1}{1+\lceil 1/\delta \rceil}} \leq 1 + \delta.\]

Again, the transformation does not change the processing order of any pieces nor the zone in which any piece is scheduled, and thus preserves the well-ordered property of Definition~\ref{def:discr}.
\end{proof}

We now show that the speedup used in our transformation does not increase the energy consumption by more than a factor of $1+\epsilon$. To this end, observe that for any $t\in \mathcal{T}$, the speed of the schedule OPT in Lemma~\ref{lem:discretized} is at most $\scrit$, by Lemma~\ref{lem:optproperties}(\ref{prop:block}). Furthermore, note that the final schedule $F$ has speed zero whenever OPT has speed zero. This allows $F$ to use exactly the same sleep phases as OPT (resulting in the same wake-up costs). It therefore suffices to prove the following lemma, in order to bound the increase in energy consumption. 

\begin{lemma}\label{lem:speedup_powerup}
For any $s \in [0,\scrit]$, there holds
\[\frac{P\big((1+\delta)^3 s\big)}{P(s)} \leq 1+\epsilon.\]
\end{lemma}
\begin{proof}
\begin{multline*}
\frac{P\big((1+\delta)^3 s\big)}{P(s)} \overset{(1)}{\leq}
\frac{P\big((1+4\delta) s\big)}{P(s)} = \frac{P(s) +4 \delta s
  \frac{P(s+4 \delta s)-P(s)}{4 \delta s} }{P(s)} \\ 
\overset{(2)}{\leq} \frac{P(s) +4 \delta s
  \frac{P(s+\scrit)-P(s)}{\scrit} }{P(s)}  \overset{(3)}{\leq}
\frac{P(s) +4 \delta s \frac{P(2\scrit)-P(\scrit)}{\scrit} }{P(s)} \\
\overset{(4)}{\leq} 1+4 \delta \frac{\scrit}{P(\scrit)} \cdot \frac{P(2\scrit)-P(\scrit)}{\scrit} \overset{(5)}{\leq} 1+\epsilon.
\end{multline*}
In the above chain of inequalities, (1) holds since
$\delta \le \frac{1}{4}$ and $P(s)$ is non-decreasing. (2) and (3)
follow from the convexity of $P(s)$, and the fact that $4 \delta s \leq
\scrit$. Inequality (4) holds since $\scrit$ minimizes $P(s)/s$ (and thus maximizes $s/P(s)$), and (5)
follows from the definition of $\delta$.
\end{proof}

We summarize the major result of this section in the following lemma.
\begin{lemma}\label{lem:nearopt}
There exists a well-ordered discretized schedule with an energy consumption no more than $(1+\epsilon)$ times the optimal energy consumption for $\J'$. \end{lemma}

\section{The Dynamic Program}
\label{sec:dp}

In this section, we show how to use dynamic programming to find a well-ordered discretized schedule with minimum energy consumption.
In the following, we discuss only how to find the minimum energy consumption of this target schedule, as the actual 
schedule can be easily retrieved by proper bookkeeping in the dynamic programming process.

Recall that $D$ is the set of all pieces and $W$ the set of time points.
Let $u_1, u_2, \ldots, u_{|D|}$ be the pieces in $D$, and w.l.o.g. assume that $u_1\prec u_2\prec \ldots \prec u_{|D|}$. 

\begin{definition} For any $k \in \{1,\ldots,|D|\}$, and $\tau_1 \leq \tau_2$, $\tau_1$, $\tau_2 \in W$, we define 
$E_k(\tau_1,\tau_2)$ as the minimum energy consumption during the interval $[\tau_1,\tau_2]$,
of a well-ordered discretized schedule so that 

\begin{enumerate}
\item all pieces $\{u\succeq u_k: \tau_1<d_u\leq \tau_2\}$ are processed in the interval $[\tau_1, \tau_2)$, and 
\item the machine is active right before $\tau_1$ and right after $\tau_2$. 
\end{enumerate}
\label{def:table}
\end{definition}

In case that there is no such feasible schedule, let $E_k(\tau_1, \tau_2) = \infty$. 

The DP proceeds by filling the entries $E_k(\tau_1, \tau_2)$ by decreasing index of $k$. The base cases are 
\begin{eqnarray*}
E_{|D|+1}(\tau_1, \tau_2) &:= & \min\{ P(0)(\tau_2-\tau_1), C\}, \forall \tau_1,\tau_2 \in W, \tau_1 \le \tau_2. \ 
%E_{k}(\tau, \tau) & = & 0, \ \ \ \ \ \ \ \ \ \ \ \ \ \ \ \ \ \ \ \ \ \ \ \ \ \ \ \forall 1 \leq k \leq |D|, \ \forall \tau \in W
\end{eqnarray*}

%\noindent (recall that $C$ is the wake-up cost). And 

%$$\forall 1 \leq k \leq |D|, \ \forall \tau \in W, \ \ E_{k}(\tau, \tau)=\left\{
%\begin{array}{c l} 
% \infty & \mbox{if $\tau$ is the deadline of any piece $u$, } \\
% 0 & \mbox{o.w.} 
%\end{array}\right.$$

For the recursion step, suppose that we are about to fill in $E_k(\tau_1,\tau_2)$. There are two possibilities. 

\begin{itemize}

\item Suppose that $d_{u_k} \not \in (\tau_1, \tau_2]$. Then clearly $E_k(\tau_1,\tau_2) = E_{k+1}(\tau_1,\tau_2)$. 

\item Suppose that $d_{u_k} \in (\tau_1, \tau_2]$. By definition, piece $u_k$ needs to be processed in the interval $[\tau_1,\tau_2)$. 
We need to guess its actual execution period $[b,e) \subseteq [\tau_1,\tau_2)$, and process the remaining pieces 
$\{u\succeq u_{k+1}: \tau_1<d_u\leq \tau_2\}$ in the two intervals $[\tau_1, b)$ and $[e,\tau_2)$. 
We first rule out some guesses of $[b,e)$ that are bound to be wrong. 

\begin{itemize}
\item By Definition~\ref{def:discr}(\ref{prop:disc1}), in a discretized schedule, a piece has to be processed completely inside a zone 
$[t_i,t_{i+1})$ (recall that $t_i \in W'$ are the release times and
deadlines of the jobs). Therefore, in the right 
guess, the interior of $[b,e)$ does not contain any release times or deadlines; more precisely, there is no 
time point $t_i \in W'$ so that $b < t_i < e$. 

\item By Definition~\ref{def:discr}(\ref{prop:disc3}), in a well-ordered discretized schedule, if piece $u_k$ ends at time point $e$, then all 
pieces $u' \succ u_k$ with deadline $d_{u'} \geq e$ are processed \emph{after} $u_k$. However, consider the guess $[b,e)$, where $e = d_{u'}$ for some $u' \succ u_{k}$ (notice that the previous case does not rule out this possibility). Then $u'$ cannot be processed 
anywhere in a well-ordered schedule. Thus, such a guess $[b,e)$ cannot be right. 
\end{itemize} 

By the preceding discussion, if the guess $(b,e)$ is right, 
the two sets of pieces $\{u\succeq u_{k+1}: \tau_1<d_u\leq b\}$ and $\{u\succeq u_{k+1}: e <d_u\leq \tau_2 \}$, along with piece $u_k$, 
comprise all pieces to be processed that are required by the definition of $E_k(\tau_1, \tau_2)$. 
Clearly, the former set of pieces $\{u\succeq u_{k+1}: \tau_1<d_u\leq b\}$ has to be processed 
in the interval $[\tau_1,b)$; the latter set of pieces, in a well-ordered schedule, 
must be processed in the interval $[e,\tau_2)$ if $[b,e)$ is the correct guess for the execution of the piece $u_k$. 

We therefore have that

$$E_k(\tau_1,\tau_2) = \min_{\substack{ b, e \in W,\ [b,e) \subseteq [\tau_1, \tau_2), \\ [b,e) \subseteq[r_{u_k}, d_{u_k}), \\ \not \exists t_i \in W', \ \mbox{s.t. } b < t_i < e, \\ \not \exists 
u' \succ u_{k}, \ \mbox{s.t. } d_{u'}=e.} }
E_{k+1}(\tau_1,b) + P(\frac{v_{u_k}}{e-b})(e-b) + E_{k+1}(e, \tau_2)$$

if there exist $b, e \in W$ with the properties stated under the min-operator, and $E_k(\tau_1,\tau_2) = \infty$, otherwise.
\end{itemize}

The correctness of the DP follows from an inductive argument. It can be verified that the running time of the DP 
is polynomial in the input size and $1/\epsilon$. The minimum energy consumption for the target schedule 
is $E_1(0, d_{max})$.

\begin{theorem}
\label{thm:main-theorem}
There exists a fully polynomial-time approximation scheme (FPTAS) for
speed scaling with sleep state.
\end{theorem}
\begin{proof}
Given an arbitrary instance $\J$ for speed scaling with sleep state,
we can transform it in polynomial time to an instance $\J'$, as seen in
Section~\ref{sec:prelim}. We then apply the dynamic programming
algorithm that was described in this section to obtain a well-ordered discretized schedule $\mathcal{S'}$
of minimal energy consumption for instance
$\J'$. By Lemma~\ref{lem:nearopt}, we have that $\mathcal{S'}$ is a $(1+\epsilon)$-approximation for instance $\J'$. Furthermore, note that every discretized schedule (and therefore also $\mathcal{S'}$) executes each task $j_i^d$ exactly in its allowed interval $I_i = [y_i, z_i)$. This holds because there are no time
points from the interior of $I_i$ included in $W$, and any discretized schedule
must therefore choose to run $j_i^d$ precisely from $y_i \in W$ to
$z_i \in W$. Therefore, by Lemma~\ref{lem:reduction}, we can transform $\mathcal{S'}$ to a schedule
$\mathcal{S}$ in polynomial time and obtain a
$(1+\epsilon)$-approximation for $\J$.
\end{proof}

%\section{Conclusion and Open Problems}
%
%The FPTAS presented, in combination with the NP-hardness result
%of~\cite{} settles the computational complexity of speed scaling with
%sleep state. A natural future direction is to study the
%approximability of the problem in a
%multiprocessor environment. Many properties that were
%crucially used in developing and analyzing both our algorithm and the preceding ones, do
%not seem to hold anymore in a multiprocessor setting. [TODO: do we
%want to keep the conclusions section or remove it completely? What are
%your opinions?]

%Can our algorithm be extended to give a simple (although probably not
%exact) algorithm for the Baptiste Chrobak Duerr setting?

\bibliography{RaceToIdle}
\bibliographystyle{plain}

\end{document}